\documentclass[letterpaper, 10 pt, conference]{ieeeconf}  
\IEEEoverridecommandlockouts             

\usepackage{bm}
\usepackage{enumerate}
\usepackage{commath}
\usepackage{graphicx}
\graphicspath{{Pictures/}}
\usepackage{amsmath}
\usepackage{subcaption}
\usepackage{breqn}
\usepackage{cite}
\usepackage[table]{xcolor}
\usepackage{booktabs}
\usepackage{empheq}
\usepackage{amsfonts}
\usepackage{amssymb}

\usepackage{amsthm}
\usepackage{hyperref}
\usepackage{xcolor}

\usepackage{mathrsfs}

\usepackage{accents}
\usepackage{thmtools}
\usepackage{thm-restate}
\usepackage{float}
\usepackage{comment}
\usepackage[normalem]{ulem}
\usepackage{algorithm}
\usepackage{algpseudocode}

\usepackage[normalem]{ulem}

\usepackage{tikz}
\newcommand*\circled[1]{\tikz[baseline=(char.base)]{
            \node[shape=circle,draw,inner sep=1pt,scale=0.8] (char) {#1};}}

\theoremstyle{plain}

\newtheorem{lemma}{Lemma}
\newtheorem{assumption}{Assumption}

\newtheorem{proposition}{Proposition}
\newtheorem*{problem*}{Problem}
\newtheorem*{theorem*}{Theorem}
\newtheorem{assumption*}{Assumption}

\theoremstyle{definition}

\newtheorem{remark}{Remark}
\newtheorem{definition}{Definition}

\definecolor{amber}{rgb}{1.0, 0.3, 0.0}

\newcommand{\myvar}[1]{#1}

\newcommand{\tildedot}[1]{\dot{\tilde{\myvar{#1}}}}

\newcommand{\myset}[1]{\mathcal{#1}} %

\algnewcommand{\Initialize}[1]{%
  \State \textbf{Initialize:}
  \Statex \hspace*{\algorithmicindent}\parbox[t]{.8\linewidth}{\raggedright #1}
}

\begin{document}

\title{
  A contract negotiation scheme for safety verification of interconnected systems 
}

\author{Xiao Tan, Antonis Papachristodoulou,  and Dimos V. Dimarogonas %
\thanks{ This work was supported in part by Swedish Research Council (VR), in part by Horizon Europe SymAware, in
part by  ERC CoG LEAFHOUND, in part by EU CANOPIES Project, and in part by  Knut and Alice Wallenberg Foundation (KAW). Xiao Tan and Dimos V. Dimarogonas are with the School of EECS, Royal Institute of Technology (KTH), 100 44 Stockholm, Sweden. Antonis Papachristodoulou is with the Department of Engineering Science, University of Oxford, Oxford, United Kingdom. Email: 
        {\tt\small \{xiaotan, dimos\}@kth.se, antonis@eng.ox.ac.uk}.}
}

\maketitle
\thispagestyle{plain}
\pagestyle{plain}

\begin{abstract}
  This paper proposes a (control) barrier function synthesis and safety verification scheme for interconnected nonlinear systems based on assume-guarantee contracts (AGC) and sum-of-squares (SOS) techniques. It is well-known that the SOS approach does not scale well for barrier function synthesis for high-dimensional systems. In this paper, we show that compositional methods like AGC can mitigate this problem. We formulate the synthesis problem into a set of small-size problems, which constructs local contracts for subsystems, and propose a negotiation scheme among the subsystems at the contract level. The proposed scheme is then implemented numerically on two examples: vehicle platooning and room temperature regulation.

\end{abstract}

\section{Introduction}

In many engineering applications, system states need to be confined to a specific set of safe states. Designing active control to achieve this property and verifying a given closed-loop system regarding this property are known as safety synthesis and verification problems. Many safety-ensuring control approaches have been proposed in the literature, including reachability analysis \cite{bansal2017hamilton}, control barrier functions (CBF) \cite{Ames2017}, model predictive control \cite{camacho2013model}, prescribed performance control \cite{bechlioulis2008robust} among many others. In particular, when a CBF is shown to exist, safety-ensuring feedback can be constructed, and the safety of the system is certified \cite{prajna2004safety}. Thus, there has been lots of interest in synthesizing valid control barrier functions numerically, by, for example, sum-of-square approaches \cite{clark2021verification,wang2023safety}, learning-based approaches \cite{robey2020learning,abate2021fossil}, and Hamiltonian-Jacobi reachability analysis \cite{choi2021robust}. However, most of these approaches are limited to dynamical systems of small to moderate size, and will become computationally intractable for large-scale systems.

Many complex, large-scale systems naturally impose an interconnected structure. It is thus essential to exploit this structure to deal with the numerical scalability issue. Along this line of research, the idea of compositional reasoning has been leveraged so that one could establish properties of the interconnected system by reasoning properties on its components. As for system safety/invariance property, \cite{jagtap2020compositional,lyu2022small} propose to synthesize local barrier functions,  establish local input-to-state safety properties, and compose the local properties by checking a small-gain-like condition. However, it remains unclear how to adapt local safety properties if the condition fails. On the other hand, \cite{coogan2014dissipativity}  certifies the safety property by seeking a Lyapunov function of the interconnected system. Safety is thus certified if a subset of the constructed Lyapunov function has no intersection with the unsafe region. It is worth noting that the search for a Lyapunov function is solved as a centralized semi-definite problem, and is still computationally demanding when the size of the interconnected system becomes larger.

In the literature of formal methods and model checking \cite{tabuada2009verification}, the composition of system properties is usually approached through the notion of an assume-guarantee contract \cite{benveniste2018contracts}. In plain words, a contract describes the behavior that a system will exhibit (guarantees) subject to the influence of the environment (assumptions). Originally, the main application domain of a contract in model checking was for discrete space systems. When contracts are applied to certify the safety of complex continuous space systems, circular reasoning of implications might exist. This is not a trivial problem in general, and the AGC framework is always sound only if a hierarchical structure exists \cite{saoud2021assume}. \cite{kim2017small} introduces parameterized AGCs, laying the foundation for finding local AGCs that can be composed of. \cite{ghasemi2020compositional} deals with invariance properties of discrete-time linear systems. The authors show that the composition of all local AGCs can be formulated as a linear program when using zonotopic representation to parameterize the constraint set and input set. In \cite{eqtami2019quantitative}, the authors consider a finite transition system and propose to determine how safe a state is by applying value iterations. The contracts are iterated locally, yet no completeness guarantee can be asserted. 

Recently, there are a few works that apply AGCs to control synthesis problems for continuous-time systems. \cite{shali2022composition} utilizes behaviour AGC for control design for linear systems, and  \cite{liu2022compositional} applies AGCs to design local feedback law under signal temporal logic specifications.

In this work, we provide a tractable safety verification scheme for continuous-time interconnected nonlinear systems, leveraging sum-of-square techniques and assume-guarantee contracts. Our result is built upon \cite{saoud2021assume} on the invariance AGCs for continuous-time systems that circumvent circular reasoning under mild assumptions. Our proposed approach consists of the construction of local AGCs and the search for compatible AGCs. In contrast to \cite{coogan2014dissipativity,ghasemi2020compositional}, we propose to negotiate local contracts only with its neighbors, and thus no central optimization is needed.  Once a set of compatible AGCs is returned, the safety of the interconnected system is certified. Moreover, we show that the proposed algorithms will terminate in finite steps and find a solution whenever one exists under relevant technical assumptions in the case of acyclic interconnections or for homogeneous systems.

\section{Notation and Preliminaries} \label{sec:notation}

\textit{Notation:}  $\mathbb{R}^n$ denotes the $n$-dimensional vector space. A vector $\myvar{a} = (a_1, a_2, \ldots, a_n)\in \mathbb{R}^n$ is a column vector unless stated otherwise. For $Z\subseteq \mathbb{R}^n$, we denote by $M(Z)$ the set of continuous-time maps $z: E \to Z$, where $E\in \{[0,a],a\ge 0\} \cup \{[0,a), a>0\} \cup \{\mathbb{R}_+\}$  is a time interval. Given sets $X_i\subseteq \mathbb{R}^{n_i}$, $i\in \myset{I} = \{1,2,\ldots,N\}$,  the Cartesian product $X_1\times X_2 \times \cdots \times X_N$ is denoted by $ \Pi_{i\in \myset{I}} X_i$. Let $x\in \mathbb{R}^n$ be an independent variable. Denote by $\mathcal{R}[x]$ the set of polynomials in the variable $x$. We call a polynomial $p\in \mathcal{R}[x]$ sum-of-squares if there exist polynomials $g_1, g_2, \ldots, g_N$ in the variable $x$ such that $p = \sum_{i= 1}^N g_i^2.$ Denote by $\Sigma [x]$ the set of sum-of-squares polynomials in $x$. Let $\mathcal{R}[x_1,x_2,\ldots, x_n], \Sigma [x_1,x_2, \ldots,x_n]$ denote the sets of polynomials and SOS polynomials of independent variables $x_1,x_2,\ldots, x_n$, respectively. Consider a directed graph $(\mathcal{I}, \mathcal{E}), \mathcal{E} \subseteq \mathcal{I}\times \mathcal{I}$. Denote by $N(i) = \{j\in \mathcal{I}: (j,i)\in \mathcal{E}\}$ the set of parent nodes of node $i$, and $\text{Child}(i) = \{k\in \mathcal{I}: (i,k)\in \mathcal{E}\}$ the set of its child nodes. We call node $i$ a root node if $N(i)=\emptyset$; node $i$ is a leaf node if $\text{Child}(i) = \emptyset$.

We first introduce the definitions of continuous-time systems, their interconnections, and assume-guarantee contracts tailored from \cite{saoud2021assume} for the safety verification problem. 

\subsection{Systems and Interconnections}
In this work, we consider continuous-time systems formally defined as follows.
\begin{definition}[Continuous-time system]
    A continuous-time system $G$ is a tuple 
    \begin{equation*}
        G = (U, W, X, Y, X^0,\mathcal{T}), 
    \end{equation*}
    where the sets $U, W, X, Y, X^0$ represent the external input set, the internal input set, the state set, the output set, and the initial state set, respectively. $\myvar{u} \in U, \myvar{w} \in W, \myvar{x} \in X, \myvar{y} \in Y$ are the external input, internal input, local state, and local output variables. $ \mathcal{T}\subseteq M(U\times W \times X \times Y)$ characterizes all the trajectories that are described by a differential equation
\begin{align} 
    \dot{\myvar{x}} (t)&= f(\myvar{x},w) + g(\myvar{x},w) \myvar{u} \label{eq:system_dynamics}
\end{align}
and $  o: \myvar{x} \mapsto \myvar{y}$ is the output function. 
\end{definition}
To guarantee the existence and uniqueness of the system trajectory, we conveniently assume that the vector field and the output map are locally Lipschitz. Now we formally define an interconnected system. 

\begin{definition}
    Given $N$ subsystems $\{ G_i\}_{i\in \myset{I}}$, $G_i = (U_i, W_i, X_i, Y_i, X^0_i,\mathcal{T}_i), \myset{I}= \{1,2,\ldots,N\},$ and  a binary connectivity relation $ \mathcal{E}\subseteq \myset{I}\times \myset{I}$, we say  $\{ G_i\}_{i\in \myset{I}}$ is \emph{compatible for composition} with respect to $\mathcal{E}$ if $\Pi_{j\in N(i)} Y_j \subseteq W_i$, where $N(i) = \{j: (j,i)\in \mathcal{E}\}$ is the index set of subsystems that influence $G_i$.  $G_j$ ($G_i$) is referred to as a parent (child) node of  $G_i$ ($G_j$).
\end{definition}
In this definition, a set of subsystems is compatible for composition when, for each subsystem, the output space of all parental subsystems is a subset of its internal input space.

When the subsystems $\{G_i\}_{i\in \myset{I}}$ are compatible for composition w.r.t.  $ \mathcal{E}$, the composed system, also referred to as the interconnected system, is denoted by $\langle (G_i)_{i\in \myset{I}}, \mathcal{E}\rangle = (U, \{0\}, X, Y, X^0, \mathcal{T})$, where $U = \Pi_{i\in \myset{I}} U_i, X = \Pi_{i\in \myset{I}} X_i, Y = \Pi_{i\in \myset{I}} Y_i,  X^0 = \Pi_{i\in \myset{I}} X_i^0$.  Denote the composed state by $\myvar{x}$, the composed external input $\myvar{u}$, and the composed output $\myvar{y}$. Then $(\myvar{u}(t),0,\myvar{x}(t),\myvar{y}(t))\in \myset{T} $ if and only if for all $i\in \myset{I}$, there exists $(\myvar{u}_i(t),\myvar{w}_i(t),\myvar{x}_i(t),\myvar{y}_i(t))\in \myset{T}_i$ and $\myvar{w}_i(t) = (\myvar{y}_{j_1}(t), \myvar{y}_{j_2}(t), \ldots , \myvar{y}_{j_p}(t)),$ where $N(i) = \{j_1, j_2, \ldots, j_p\}$.

\subsection{Assume-guarantee contracts for invariance}
To begin with, we introduce notation that will help us define the set of all continuous trajectories that always stay in a set. Let a nonempty set $S\subseteq \mathbb{R}^n$.  Define 
\begin{equation}
    I_S^E = \{z: E \to \mathbb{R}^n\in M(\mathbb{R}^n): \forall t\in E, z(t)\in S \},
\end{equation}
where $E$ is a time interval. In the following, the superscript $E$ is neglected as it is usually chosen as the maximal time interval of the existence of solutions to the continuous-time system. An invariance assume-guarantee contract (iAGC) is defined as follows:
\begin{definition} \label{def:iAGC}
    For  a continuous-time system $ G= (U, W, X, Y, X^0, \mathcal{T}) $, an \emph{invariance assume-guarantee contract} (iAGC) for $G$ is a tuple $C = ( I_{\underline{W}}, I_{\underline{X}}, I_{\underline{Y}})$ where $\underline{W}\subseteq W,\underline{X}\subseteq X,\underline{Y}\subseteq Y $. We refer to $ I_{\underline{W}}$ as the set of assumptions on the internal inputs, and $I_{\underline{X}}, I_{\underline{Y}}$ as the sets of guarantees on the states and outputs.
    We say a system $G$ \emph{satisfies a contract} $C= ( I_{\underline{W}}, I_{\underline{X}}, I_{\underline{Y}})$, denoted $G \models C$, if there exists a feedback control $k(\cdot,\cdot):X\times W \to U $ such that   
    for all $t>0$, for all $ w|_{[0,t]} \in I_{\underline{W}}$,  the state and output fulfill $x|_{[0,t]} \in I_{\underline{X}},y|_{[0,t]} \in I_{\underline{Y}}$ for all trajectories $(u(t)=k(x,w),w(t),x(t),y(t)) \in \myset{T} $.
\end{definition}

A key result that establishes the compositional reasoning of the system property is the following:
\begin{lemma}[Compositional reasoning] \label{lem:compositional_reasoning}
    Consider an interconnected system $\langle (G_i)_{i\in \myset{I}}, \mathcal{E}\rangle = (U,\{0\}, X, Y, X^0, \mathcal{T})$ composed of $N$ subsystems with a compatible binary connectivity relation $\mathcal{E}$. If for each subsystem $G_i = ( U_i, W_i, X_i, Y_i, X_i^0, \mathcal{T}_i) $, there exists an invariance assume-guarantee contract $C_i = (  I_{\underline{W}_i}, I_{\underline{X}_i},  I_{\underline{Y}_i} )$ such that $G_i \models C_i $ and $         \Pi_{j\in N(i)}{I_{\underline{Y}_j}} \subseteq I_{\underline{W}_i},$
     then $\langle (G_i)_{i\in \myset{I}}, \mathcal{E}\rangle \models C$ with $C= ( \{0\}, \Pi_{i\in \myset{I}}{I_{\underline{X}_i}}, \Pi_{i\in \myset{I}}{I_{\underline{Y}_i}}).$
\end{lemma}
This lemma  is a special case of \cite[Theorem 3]{saoud2021assume} and thus we neglect its proof here. While this lemma may seem intuitive, it is worth highlighting that we can not deduce directly the conclusion due to possible circular reasoning of implications. One such example for systems with non-locally Lipschitz vector fields is shown in \cite[Example 6]{saoud2021assume}. The AGC framework helps to circumvent possible  circular reasoning and enables compositional reasoning of the forward invariance property of interconnected systems.

\subsection{System safety and barrier functions}
Now we give the formal definition of safety of a continuous-time system. Throughout this work, we refer to a \emph{safe region} as the collection of states that are benign, for example, unoccupied configuration space in robotic applications, and a \emph{safe set} as a subset of the safety region that is also forward invariant.

\begin{definition} \label{def:system_safety}
    Given a system $G = (U, W, X, Y, X^0, \mathcal{T})$ and a safe region $\myset{Q}\subseteq X$, we say the system $G$ is \emph{safe with respect to an internal input set $\underline{W}$} if and only if $X^0 \subseteq \myset{Q}$, and for all initial states $x_0 \in X^0$ and for all internal input signals $ w|_{[0,t]}\in I_{\underline{W}},  $ there exists an external input signal $u|_{[0,t]}\in I_U$ such that $x|_{[0,t]}\in I_{\myset{Q}}$ for all $t>0$.
\end{definition}
 We note that the internal input $w$ is assumed to be known via communication. This is in contrast to the definition of a robust control invariant set where $w$ is treated as disturbance and unknown \cite{blanchini1999set}, in which case the qualifier $\exists u $  proceeds $\forall w$. One way to certify the safety of the system is to find a (control) barrier function, also known as a barrier certificate \cite{prajna2004safety}, which is given by
\begin{definition} \label{def:cbf}
    A  differential function $h: X\to \mathbb{R}$ for system $G= (U, W, X, Y, X^0, \mathcal{T})$ is called a \emph{control barrier function}  with respect to an internal input set $\underline{W}\subseteq W$, if  there exists a class $\mathcal{K}$ function $\alpha$ such that $  \forall w \in \underline{W}, \exists u\in U$ the following inequality holds for all $ x $
        \begin{equation} \label{eq:cbf_condition}
            \nabla h(x) f(x,w) + \nabla h(x)g(x,w)u + \alpha(h(\myvar{x}))\geq 0.
        \end{equation}
\end{definition}
When the external input set is empty, i.e., $U = \emptyset$, $h$ is called a \emph{barrier function} as this system has no active control. When such a (control) barrier function is found,  then the set $ \myset{C}=  \{x: h(x)\geq 0\}$ is (controlled) forward invariant, and asymptotically stable when it is compact\cite{Ames2017}. If $X^0 \subseteq \myset{C}\subseteq \myset{Q}$, then system safety is certified\cite{Ames2017,wang2023safety}. In general, finding an invariant set $\myset{C}$ is computationally expensive for large nonlinear systems. 

\subsection{Sum-of-squares programs}
One tractable approach to deal with infinite inequalities as in \eqref{eq:cbf_condition} is via sum-of-squares programming. A standard sum-of-squares (SOS) program takes the following form
\begin{equation}
    \begin{aligned}
       &  \min_{p_1,\ldots,p_k} \textstyle \sum_{j=1}^k a_j p_j \\
       & \text{s.t. } b_0(x) + \textstyle \sum_{j=1}^k p_i b_j(x) \in \Sigma[x],
    \end{aligned}
\end{equation}
where the decision variables $p_1,\ldots,p_k\in \mathbb{R}$,  constants $a_1,\ldots,a_k\in \mathbb{R}$ are the weights, and $b_0,\ldots, b_k\in \mathcal{R}[x]$ are given polynomials. This SOS program is a convex optimization problem and can be equivalently transformed into a semi-definite program (SDP). Interested readers are referred to \cite{prajna2002introducing,arcak2016networks} for more details.

\subsection{Problem formulation}
In this work, we aim to numerically verify the safety property of interconnected systems. The following sub-problems are considered:
\begin{enumerate}
    \item[(P1)] For a continuous-time subsystem $ G_i = ( U_i, W_i, X_i, Y_i, X^0_i, \mathcal{T}_i) $ and a given safe region $\myset{Q}_i\subseteq X_i$, construct an invariance assume-guarantee contract $C_i = ( I_{\underline{W}_i}, I_{\underline{X}_i}, I_{\underline{Y}_i})$ such that $G_i \models C_i$ and $X_i^0\subseteq \underline{X}_i\subseteq \myset{Q}_i$;
    \item[(P2)] For an interconnected system $\langle (G_i)_{i\in \myset{I}}, \mathcal{E}\rangle=(U, \{0\},X, Y, X^0, \mathcal{T}) $ and a safe region $\myset{Q} = \Pi_{i\in \myset{I}}\myset{Q}_i, $ $\myset{Q}_i \subseteq X_i$, construct an invariance contract $C = ( \{0\}, I_{\underline{X}}, I_{\underline{Y}})$ such that $G \models C$ and $X^0\subseteq \underline{X}\subseteq \myset{Q}$.
\end{enumerate}
If such a $\underline{X}$ is found, then safety of the interconnected system is certified.

\begin{assumption} \label{ass:ass1}
We assume the following:
\begin{enumerate}
    \item the local feedback law $u_i = k_i(x_i,w_i) \in U_i$ is known, but it does not necessarily render the interconnected system safe;
    \item The class $\mathcal{K}$ function $\alpha(\cdot)$ in \eqref{eq:cbf_condition} is chosen to be a linear function with constant gain $a$.
    \item The initial set $X_i^0$, safe region $\myset{Q}_i$, and the internal input set $W_i$ are super-level sets of (possibly vector-valued) differentiable functions, i.e., $X_i^0 = \{x_i: b_i^0(x_i)\geq 0\}, \myset{Q}_i =\{x_i: q_i(x_i)\geq 0\}, W_i =\{(\myvar{y}_{j_1}, \myvar{y}_{j_2}, \ldots, \myvar{y}_{j_p}): d^i_{j_k}(y_{j_k})\geq 0, k = 1,2,\ldots, p\}, \text{where } N(i) = \{j_1, j_2, \ldots, j_p\} $.
    \item  $b_i^0,q_i\in \mathcal{R}[x_i], d^i_{j_k}(y_{j_k}) \in \mathcal{R}[y_{j_k}], f_i, g_i,k_i \in \mathcal{R}[x_i,w_i]$ are polynomials.
    \item The subsets of $W_i, \myset{Q}_i$, i.e., $\underline{W}_i, \underline{\myset{Q}}_i$ are chosen in the form of
\begin{equation*}
\begin{aligned}
& \underline{\myset{Q}}_i = \{x_i: q_i(x_i)\geq \zeta \bm{1} \text{ for some } \zeta\geq 0 \},  \\
 &    \underline{W}_i = \{ (\myvar{y}_{j_1},  \ldots , \myvar{y}_{j_p}): d_{j_k}^i(y_{j_k})\geq \delta \bm{1} \text{ for some } \delta\geq 0 \}.   
\end{aligned}
\end{equation*} 
\item When searching for non-negative polynomials, we restrict the search to the set of SOS polynomials up to a certain degree.
\end{enumerate}

\end{assumption}
These restrictions, even though conservative, will facilitate the convergence and completeness analysis that will be clear later. We note that the first two assumptions are in place to avoid bilinear terms when constructing the SOS programs. Both can be relaxed by considering iterative optimization approaches. See \cite{wang2023safety} for more details. Assumptions \ref{ass:ass1}.3 and \ref{ass:ass1}.5 help to parameterize the assumption and the guarantee sets by scalar variables $\delta$ and $\zeta$, respectively. Assumptions \ref{ass:ass1}.4 and \ref{ass:ass1}.6 are standard in the field of SOS-based system verification.

For notational simplicity, we define the set projection of an internal input set $W_{i}$ of subsystem $G_i$ with respect to subsystem $G_k$ as $\text{Proj}_{k}(W_{i}) = \{ y_{k}:d^i_{k}(y_{k})\geq 0\} $ if $k \in N(i)$, and $\text{Proj}_{k}(W_{i}) = \emptyset$ otherwise.  For a mapping $o: X \to Y$, let $o^{-1}(\underline{Y}) := \{x: o(x)\in \underline{Y}\}$.

\section{Proposed solutions} \label{sec:solution}
The proposed approach consists of 1) numerically constructing iAGCs for subsystems by synthesizing local (control) barrier functions, and 2) negotiating iAGCs among subsystems to certify the safety property of the interconnected system. We also discuss the convergence properties of our approach.

\subsection{Local barrier function and AGC construction}
In this subsection, we will focus on tackling Problem (P1) for a subsystem $G_i = ( U_i, W_i, X_i, Y_i, X^0_i, \mathcal{T}_i) $. Under Assumption \ref{ass:ass1}, the closed-loop subsystem dynamics of \eqref{eq:system_dynamics} are 
\begin{equation} \label{eq:local_dyn}
     \dot{\myvar{x}}_i (t)= f_i(\myvar{x}_i,\myvar{w}_i) + g_i(\myvar{x}_i,\myvar{w}_i) k_i(\myvar{x}_i,\myvar{w}_i):=F_i(\myvar{x}_i,\myvar{w}_i).
\end{equation}

In this subsection, for the sake of notation simplicity, we will drop the subscript $i$ when no confusion arises.

First we show the relations between a) finding a control barrier function, b) constructing an invariance assume-guarantee contract, and c) establishing the safety property of a subsystem.
\begin{proposition} \label{prop:cbf_as_iAGC}
    Consider a continuous-time system $G = ( U, W, X, Y, X^0, \mathcal{T}) $,  a safe region $\myset{Q}$ and an internal input set $\underline{W}\subseteq W$. Consider the following claims:
    \begin{enumerate}
        \item[\circled{1}]  there exists a CBF $h$ with respect to the internal input set $\underline{W}$. Denote by $\myset{C} = \{x: h(x)\geq 0\}$;
        \item[\circled{2}]  the system $G\models C$, where $C = ( I_{\underline{W}}, I_{\myset{C}}, I_{o(\myset{C})})$;
        \item[\circled{3}] $X^0 \subseteq \myset{C}\subseteq \myset{Q}$;
        \item[\circled{4}] the system is safe with respect to $\underline{W}$;
    \end{enumerate}
    We have \circled{1} $\implies $ \circled{2};   \circled{2} and \circled{3} $\implies $  \circled{4}.
\end{proposition}
\begin{proof}
     \circled{1} $\implies $ \circled{2}: \circled{1} implies that the set $\myset{C}$ is forward invariant when the signal $w(t)\in I_{\underline{W}}$, which implies \circled{2} from Definition \ref{def:iAGC}.    \circled{2} and \circled{3} $\implies $  \circled{4}: This is straightforward according to Definition \ref{def:system_safety}.
\end{proof}

Numerically, one can formulate the conditions of \circled{1} and \circled{3} of Proposition \ref{prop:cbf_as_iAGC} as a set of SOS constraints, as follows.
\begin{proposition} \label{prop:local_constract}
    Consider a continuous-time system $G= (U, W, X, Y, X^0, \mathcal{T})$ and a safe region $\myset{Q}$. If there exist SOS polynomials $\sigma_{init},\sigma_{safe}\in \Sigma[\myvar{x}]$, $ \sigma_{k} \in \Sigma[x,\myvar{y}_{k}], k = 1,2,\ldots,p$, polynomial $h\in \mathcal{R}(\myvar{x}) $, and  positive $\epsilon, a, \delta $ such that
    \begin{subequations} \label{eq:local_sos}
    \begin{align}
        h(x) - \sigma_{init} b^0(x) & \in \Sigma[x]; \label{eq:sos_initialset} \\
        -h(x) + \sigma_{safe}q(x) & \in \Sigma[x]; \label{eq:sos_safeset} \\
         \nabla h(\myvar{x}) F(\myvar{x},y_1,\ldots,y_p) + ah(\myvar{x}) \hspace{0.3cm}  & \notag \\
         - \sum_{k= 1}^p \sigma_{k} (d_k(y_k)  - \delta) -\epsilon & \in \Sigma[\myvar{x}, y_1,\ldots,y_p]. \label{eq:sos_cbf_condiiton}
    \end{align}   
    \end{subequations}
    then, letting $\underline{W} = \{ (\myvar{y}_{1}, \ldots,\myvar{y}_{k} \ldots, \myvar{y}_{p}): d_{k}(y_{k})\geq \delta \}$,   \circled{1}, \circled{2}, \circled{3}, \circled{4} in Proposition \ref{prop:cbf_as_iAGC} hold.
\end{proposition}
\begin{proof}
    \eqref{eq:sos_cbf_condiiton} is a SOS polynomial and thus $\nabla h(\myvar{x}) F(\myvar{x},\myvar{w}) + ah(\myvar{x})   \geq 0, \forall w\in \underline{W}$. This shows that claim \circled{1} holds. Based on non-negativeness of SOS polynomials, \eqref{eq:sos_initialset} and \eqref{eq:sos_safeset} imply that $\forall x, b^0(x)\geq 0 \implies h(x)\geq 0$,  and $\forall x, h(x)\geq 0 \implies q(x)\geq 0$, respectively. That is, $X_i^0 \subseteq \myset{C} \subseteq \myset{Q}$. Thus  \circled{3} holds. Following Proposition \ref{prop:cbf_as_iAGC}, we conclude the proof.
\end{proof}

 Even though \eqref{eq:local_sos} is only a sufficient condition for system safety, it is a condition we can verify numerically (and efficiently when the system size is small). For this reason, we say that $G$ is \emph{certified to be safe} in $\myset{Q}$ w.r.t. $\underline{W}$ if condition \eqref{eq:local_sos} holds. We introduce the following special sets that are useful for contract composition later. In what follows, we take $0<\epsilon<<1$ and $ a $ in \eqref{eq:local_sos} to be positive constants.

\subsubsection{Maximal internal input set} 
To quantify the largest internal input set a subsystem can tolerate while still remaining safe, we propose the following optimization problem:
\begin{equation} \label{eq:maximal internal input set}
    \begin{aligned}
        & \min \delta \\
       \text{s.t. }  & \eqref{eq:sos_initialset}, \eqref{eq:sos_safeset}, \eqref{eq:sos_cbf_condiiton}, \delta\geq 0,
    \end{aligned}
\end{equation}
where the decision variables include SOS polynomials $\sigma_{init},\sigma_{safe}\in \Sigma[\myvar{x}]$,  $ \sigma_{k} \in \Sigma[\myvar{y}_{k}], k = 1,2,\ldots,p$, polynomials $h\in \mathcal{R}(\myvar{x}) $, and a scalar $\delta$. It should be noted that although \eqref{eq:maximal internal input set} contains a bilinear term $\sigma_{input} \delta$, this can be solved efficiently by bisection as $\delta$ is a scalar. If \eqref{eq:maximal internal input set} is feasible, denote the optimal value by $\delta^\star$ and the corresponding internal input set $\underline{W}^\star$. We call $\underline{W}^\star$ the \emph{maximal internal input set} for a given subsystem $G$ and safe region $\myset{Q}$.

\subsubsection{Minimal safe region}
Given a subsystem $G$ with an internal input set $\underline{W}$, to quantify the least impact on its child subsystem, we propose the following optimization problem:
\begin{equation} \label{eq:minimal safe set}
    \begin{aligned}
        & \max \zeta \\
       \text{s.t. }  & \eqref{eq:sos_initialset}, \eqref{eq:sos_cbf_condiiton}, \zeta\geq 0 \\
       & -h(x) + \sigma_{safe}(q(x) -\zeta) \in \Sigma[x];
    \end{aligned}
\end{equation}
where the decision variables include SOS polynomials $\sigma_{init},\sigma_{safe}\in \Sigma[\myvar{x}]$, $ \sigma_{k} \in \Sigma[\myvar{y}_{k}], k = 1,2,\ldots,p$, polynomials $h\in \mathcal{R}(\myvar{x}) $, and a scalar $\zeta$. We take $\epsilon, a$ to be positive constants. $\delta$ in \eqref{eq:sos_cbf_condiiton} is known as we assume $\underline{W}$ is given. It should be noted that although \eqref{eq:minimal safe set} contains a bilinear term $\sigma_{safe} \zeta$, this can be solved efficiently by bisection as $\zeta$ is a scalar. If \eqref{eq:minimal safe set} is feasible, denote the optimal value by $\zeta^\star$ and the corresponding safe region $\underline{\myset{Q}}^\star$. We call $\underline{\myset{Q}}^\star$ the \emph{minimal safe region} for a given $\underline{W}$.

We have the following properties about the maximal internal input set $\underline{W}^\star$ and the corresponding minimal safe region $\underline{\myset{Q}}^\star$. 
\begin{proposition} \label{prop:local_sets}
    Under Assumption \ref{ass:ass1}, for  a continuous-time system $G= (U, W, X, Y, X^0, \mathcal{T})$ and a safe region $\myset{Q}$, the following results hold:
    \begin{enumerate}
        \item If \eqref{eq:maximal internal input set} is feasible for some $\delta^\prime\geq 0$, then \eqref{eq:maximal internal input set} is also feasible for $\delta'', \delta'' \geq  \delta^\prime$. If \eqref{eq:minimal safe set} is feasible for some $\zeta^\prime >0$, then \eqref{eq:minimal safe set} is also feasible for $\zeta'', 0\leq \zeta'' \leq \zeta^\prime$.
        \item Consider two safe regions $\underline{\myset{Q}}^\prime \subseteq \underline{\myset{Q}}'' \subseteq \myset{Q} $. If \eqref{eq:maximal internal input set} is feasible for the safe region $\underline{\myset{Q}}^\prime $, then \eqref{eq:maximal internal input set} is also feasible for  $\underline{\myset{Q}}'' $. Denoting the respective optimal values by $\delta^\prime, \delta''$ and the corresponding internal input sets $\underline{\myset{W}}^\prime, \underline{\myset{W}}'' $, then $\delta'' \leq  \delta^\prime$ and $\underline{\myset{W}}^\prime \subseteq \underline{\myset{W}}'' \subseteq \myset{W} $.
        \item Consider two internal input sets $\underline{\myset{W}}^\prime \subseteq \underline{\myset{W}}'' \subseteq \myset{W} $. If \eqref{eq:minimal safe set} is feasible with the internal input set $\underline{\myset{W}}'' $, then \eqref{eq:minimal safe set} is also feasible for  $\underline{\myset{W}}^\prime $. Denoting the respective optimal values by $\zeta^\prime, \zeta''$ and the corresponding safe regions $\underline{\myset{Q}}^\prime, \underline{\myset{Q}}'' $, then $  0 \leq  \zeta'' \leq \zeta^\prime $ and  $ \underline{\myset{Q}}^\prime\subseteq  \underline{\myset{Q}}'' \subseteq \myset{Q} $.
        \item If \eqref{eq:maximal internal input set} is feasible, then $\underline{W}^\star$ is the largest internal input set w.r.t. which $G$ is certified to be safe; if infeasible, then there exists no $\underline{W}\subseteq W$ w.r.t. which  $G$ can be certified to be safe.
        \item If \eqref{eq:maximal internal input set} is feasible, letting $\underline{W} = \underline{W}^\star$, then \eqref{eq:minimal safe set} is feasible and $\underline{\myset{Q}}^\star$ is the smallest safe region in which  $G$ is safe  w.r.t.  $\underline{W}^\star$.
    \end{enumerate}
\end{proposition}
\begin{proof}
    Claim 1) holds since one verifies that, when $\delta'' \geq  \delta^\prime$, any feasible decision variables $\sigma_{init},\sigma_{safe},\sigma_{k}, k = 1,2,\ldots,p, h$ to \eqref{eq:maximal internal input set} for $\delta^\prime$ are also feasible for the case of $ \delta''$. Similarly, when $0\leq \zeta'' \leq \zeta^\prime$, any feasible decision variables to \eqref{eq:minimal safe set} for $\zeta^\prime$ are also feasible for the case of $\zeta''$.

   For $\underline{\myset{Q}}^\prime \subseteq \underline{\myset{Q}}''$ (the corresponding $\zeta^\prime \geq  \zeta'' \geq 0$), if \eqref{eq:maximal internal input set} is feasible for $\zeta^\prime$, then the feasible decision variables are also feasible for the case of $ \zeta''$. As \eqref{eq:maximal internal input set} minimizes over $\delta$ and the feasible set of the case $\zeta^\prime $ is a subset of that of $\zeta'' $, then  $\delta'' \leq \delta^\prime$, proving Claim 2). 

    Similar argument applies for Claim 3).   
    For $\underline{\myset{W}}^\prime \subseteq \underline{\myset{W}}'' \subseteq \myset{W} $ (with the corresponding $\delta^\prime \geq \delta''\geq 0$),  if \eqref{eq:minimal safe set} is feasible for $\delta''$, then the feasible decision variables are also feasible for the case of $ \delta^\prime$. As \eqref{eq:minimal safe set} maximizes over $\zeta$ and the feasible set of the case $\delta'' $ is a subset of that of $\delta^\prime $, then  $\zeta^\prime \geq  \zeta''$, proving Claim 3). 
    
    Claim 4) is true since $  \underline{W}'' \subseteq \underline{W}^\prime  $ if and only if the corresponding $\delta''  \geq \delta^\prime $, and the program in \eqref{eq:maximal internal input set} minimizes $\delta$. We note that the condition $\delta\geq 0$ comes from the condition $\underline{W}\subseteq W$. If \eqref{eq:maximal internal input set} is feasible, then the feasible decision variables for $\delta = \delta^\star$ are also feasible for \eqref{eq:minimal safe set} when $\zeta = 0$. Further noting that  $  \underline{Q}'' \subseteq \underline{Q}^\prime  $ if and only if the corresponding $\zeta''  \geq \zeta^\prime $ and \eqref{eq:minimal safe set} maximizes $\zeta$, thus  Claim 5) is deduced.
\end{proof}
Proposition \ref{prop:local_sets}'s items 2 and 3 show a monotonic relation between the internal input sets and the safe regions.  Intuitively, with a larger safe region, the system can tolerate a larger disturbance (internal input set); with a larger disturbance (internal input set), the most confined safe region will become larger. Proposition \ref{prop:local_sets}'s items 4 and 5 further state that,  
for a given safe region, $\underline{W}^\star$ is the largest internal input set that a system can bear while remaining safe; for a given internal input set, $\underline{\myset{Q}}^\star$ is the most confined influence a system has for its child subsystems. When \eqref{eq:maximal internal input set} and \eqref{eq:minimal safe set} are feasible for a subsystem $G_i$, denoting the corresponding sets as $\underline{W}_i^\star, \underline{\myset{Q}}_i^\star$, then we construct a local contract $C_i = (  I_{\underline{W}_i}, I_{\underline{X}_i},  I_{\underline{Y}_i} )$, where $\underline{W}_i = \underline{W}_i^\star,\underline{X}_i = \underline{\myset{Q}}_i^\star$, and $\underline{Y}_i = o_i(\underline{X}_i)$.

\subsection{Contract composition and negotiation}
In this section, we consider the interconnected system $G = \langle (G_i)_{i\in \myset{I}}, \mathcal{E}\rangle$,  $G_i = ( U_i, W_i, X_i, Y_i, X^0_i, \mathcal{T}_i) $ with safe region $\myset{Q}_i\subseteq X_i$. We have the following results on the safety properties of the interconnected system.
\begin{proposition} \label{prop:iAGC_composition}
    If, for each subsystem $G_i $,  an iAGC $C_i =  (  I_{\underline{W}_i}, I_{\underline{X}_i},  I_{\underline{Y}_i} )$ exists such that $ X_i^0 \subseteq \underline{X}_i \subseteq \myset{Q}_i$ and
    \begin{equation} \label{eq:agreeing_condition}
        \Pi_{j\in N(i)} \underline{Y}_j \subseteq \underline{W}_i,
    \end{equation}
    then the interconnected system $\langle (G_i)_{i\in \myset{I}}, \mathcal{E}\rangle$ is safe.  
\end{proposition}
\begin{proof}
    From Lemma \ref{lem:compositional_reasoning},  $\langle (G_i)_{i\in \myset{I}}, \mathcal{E}\rangle \models C$ with $C= ( \{0\}, \Pi_{i\in \myset{I}}{I_{\underline{X}_i}}, \Pi_{i\in \myset{I}}{I_{\underline{Y}_i}}).$ Note that $ \Pi_{i\in \myset{I}}{\underline{X}_i^0} \subseteq \Pi_{i\in \myset{I}}{\underline{X}_i} \subseteq \Pi_{i\in \myset{I}}{\myset{Q}_i}$, thus $\langle (G_i)_{i\in \myset{I}}, \mathcal{E}\rangle $ is safe.
\end{proof}
We refer to the condition \eqref{eq:agreeing_condition} as the \emph{contract compatibility condition} as it indicates whether the contract of a subsystem agrees with that of its parent subsystems. In the general case, the contracts $C_i, i \in \mathcal{I}$ found locally may not satisfy this condition, and we have to refine them so that \eqref{eq:agreeing_condition} holds. We call this refinement process \emph{negotiation}. In what follows, we consider three cases and propose several different algorithms. We note that all algorithms are sound, but differ in finite-step termination and completeness guarantees.

\vspace{1mm}
\subsubsection{Acyclic  connectivity graph}
In this case, we assume that there exists no cycle in the connectivity graph $(\mathcal{I},\mathcal{E})$. In this case, the hierarchical tree structure resembles a client-contractor relation model. For $k\in \text{Child}(i)$, we could view $G_k$ as a client with an iAGC $(  I_{\underline{W}_k}, I_{\underline{X}_k},  I_{\underline{Y}_k} )$, who gives specifications on the behaviour of its parent node $G_i$ (viewed as contractors) by $\underline{W}_k$. Based on this interpretation,  we propose Algorithm \ref{alg:acyclic_graph}. 

In Algorithm \ref{alg:acyclic_graph},  $ \mathcal{I}_0, \mathcal{I}_1, \mathcal{I}_{-1}$ represent the index sets of ready-to-update, to-be-updated, and updated subsystems, respectively. The algorithm starts with the local contract construction for the leaf nodes. Following a bottom-up traversal along the connectivity graph, for each subsystem $G_i$ in $\mathcal{I}_0$, Algorithm \ref{alg:acyclic_graph} first updates its safe region $\myset{Q}_i$ such that it agrees with all its child nodes. This is explicitly conducted in Algorithm \ref{alg:update safe region}, while no operation is needed for leaf nodes. The set intersection in Algorithm \ref{alg:update safe region} is again cast as a SOS program, as follows:
\begin{equation} \label{eq:alg_safe_region}
\begin{aligned}
       &  \hspace{20mm} \min_{ \zeta \geq 0} \zeta \\
       \text{ s.t. } & q_i(x_i) - \zeta - \sigma_{k} (d^k_i \circ o_i(x_i) - \delta^k) \\
       & \hspace{23mm}\in \Sigma[x_i], \forall k \in \text{Child}(i),
\end{aligned} 
\end{equation}
where the decision variables include $\sigma_k \in \Sigma[x_i], k\in \text{Child}(i), $ and a scalar $\zeta$. Recall here $o_i$ is the output map of subsystem $G_i$, $\text{Proj}_{i}(\underline{W}_{k}) = \{y_i: d^k_i(y_i) \geq  \delta^k\}$. Denoting the optimal value by $\zeta^\prime$ and $\myset{Q}_i^\prime = \{x: q_i(x_i) \geq\zeta^\prime \}$, $\myset{Q}_i^\prime$ is then the largest inner-approximation of $ \bigcap_{k\in \text{Child}(i)} o_i^{-1}(\text{Proj}_{i}(\underline{W}_{k})) \cap \myset{Q}_i$. Recall that the subset of $\myset{Q}_i$ is parameterized by $\zeta$ from Assumption \ref{ass:ass1}.5.

After updating the safe region, Algorithm \ref{alg:acyclic_graph} calculates the maximal internal input set $\underline{W}_i^\star$ (Line 6), which can be seen as the least requirement on its parent nodes as discussed in Proposition \ref{prop:local_sets}. Algorithm \ref{alg:update I_0 and I_1} then moves $G_i$ to $\mathcal{I}_{-1}$, and checks for every to-be-updated subsystems whether all their child subsystems have been updated. If yes, then that subsystem is moved to the set of ready-to-update subsystems $\mathcal{I}_0$ and will be updated accordingly.

\begin{proposition} \label{prop:acyclic graph}
Consider an interconnected system with an acyclic connectivity graph $(\mathcal{I},\mathcal{E})$. Algorithm \ref{alg:acyclic_graph} has the following properties:
\begin{enumerate}
    \item Algorithm \ref{alg:acyclic_graph} terminates in finite steps and returns either \texttt{True} or \texttt{False}.
    \item If Algorithm \ref{alg:acyclic_graph} returns \texttt{True}, then  iAGCs $C_i =  (  I_{\underline{W}_i}, I_{\underline{X}_i},  I_{\underline{Y}_i} ), i\in \mathcal{I}$ satisfy the conditions in Proposition \ref{prop:iAGC_composition}.
    \item If Algorithm \ref{alg:acyclic_graph} returns \texttt{False}, then there exist no iAGCs $C_i =  (  I_{\underline{W}_i}, I_{\underline{X}_i},  I_{\underline{Y}_i} ), i\in \mathcal{I}$ that satisfy the conditions in Proposition \ref{prop:iAGC_composition}  under Assumption \ref{ass:ass1}.
\end{enumerate}
\end{proposition}

\begin{proof}
    For every iteration of Algorithm \ref{alg:acyclic_graph}, it will either terminate due to \texttt{infeasibility} or increase the cardinality of $\myset{I}_{-1}$ by $1$. Since $|\myset{I}_{-1}|$ is upper bounded by $|\myset{I}|$, we know it has to terminate in finite steps. This proves Claim 1).
    
    When Algorithm \ref{alg:acyclic_graph} returns \texttt{True},  each subsystem has gone through Step 4 - Step 10 and computed an iAGC $(  I_{\underline{W}_i^\star}, I_{\underline{X}_i},  I_{\underline{Y}_i} ) $ in a bottom-up transverse. As Step 4 reduces the safe region of subsystem $G_i$, and from \eqref{eq:maximal internal input set}, we have $ X_i^0 \subseteq \underline{X}_i \subseteq \myset{Q}_i^\prime \subseteq \myset{Q}_i$. Moreover, for any $k\in \text{Child}(i)$, based on Algorithm \ref{alg:update safe region}, $ \underline{Y}_i  = o_i(\underline{X}_i) \subseteq o_i(\myset{Q}_i^\prime) \subseteq \text{Proj}_{i}(\underline{W}_{k})$, which proves that the contract compatibility condition \eqref{eq:agreeing_condition} holds. Following Proposition \ref{prop:iAGC_composition}, we thus certify the safety of the interconnected systems. This proves Claim 2).

    We show claim 3) by contradiction. Suppose that there exist iAGCs $C_i =  (  I_{\underline{W}_i}, I_{\underline{X}_i},  I_{\underline{Y}_i} ), i\in \mathcal{I}$ that satisfy the conditions in Proposition \ref{prop:iAGC_composition}. Denote the safe regions with which such iAGCs are obtained by $\myset{Q}_i^\prime, i\in \mathcal{I}$, i.e., the functions defining those sets fulfill the SOS constraints in \eqref{eq:minimal safe set} (but not necessarily minimize the size of the safe region) and the sets fulfill the compatibility condition \eqref{eq:agreeing_condition}. Thanks to the tree structure, we can start our argumentation from the leaf nodes and iteratively reason about the nodes that are one level above, and end at the root node. From Proposition \ref{prop:local_sets}, for $G_i$ being the leaf nodes, we have $\underline{W}_i\subseteq \underline{W}_i^\star$, the set obtained from \eqref{eq:maximal internal input set}. For $G_j$ being the nodes one level above the leaf nodes, from Algorithm \ref{alg:update safe region}, we know $\myset{Q}_j^\prime \subseteq \myset{Q}_j^{\prime\star}$, for that $\myset{Q}_j^{\prime\star}$ is the largest inner-approximation of all safe regions and that $\underline{W}_i\subseteq \underline{W}_i^\star$. Following  Proposition \ref{prop:local_sets} item 2, we know  $\underline{W}_j \subseteq \underline{W}_j^\star $, where $\underline{W}_j^\star$ is the set obtained by solving \eqref{eq:maximal internal input set} with $\myset{Q}_j^{\prime\star}$. Recursively, we thus obtain that $(  I_{\underline{W}_i^\star}, I_{\underline{X}_i},  I_{\underline{Y}_i} ), i\in \mathcal{I}$ exists. This contradicts with the premise that Algorithm \ref{alg:acyclic_graph} returns \texttt{False}. Thus Claim 3) is proven. 
\end{proof}

\begin{algorithm}[h]
\caption{\texttt{Contract construction for acyclic graph}} \label{alg:acyclic_graph}
 \algblock[TryCatchFinally]{try}{endtry}
\algcblock[TryCatchFinally]{TryCatchFinally}{finally}{endtry}
\algcblockdefx[TryCatchFinally]{TryCatchFinally}{catch}{endtry}
	[1]{\textbf{catch} #1}
	{\textbf{end try}}
 
\begin{algorithmic}[1]
\Require  $G_i$, $\myset{Q}_i, \forall i\in \mathcal{I} $
\State $ \mathcal{I}_0  \leftarrow $ set of leaf nodes, $\mathcal{I}_1  \leftarrow \mathcal{I}\setminus \mathcal{I}_0 $, $\mathcal{I}_{-1}\leftarrow \emptyset$.
\While{$\mathcal{I}_0 \neq \emptyset$}
\For{each subsystem $ G_i, i \in  \mathcal{I}_0 $}
\State $\myset{Q}_i^\prime \leftarrow$ update the local safe region $\myset{Q}_i$ by Alg. \ref{alg:update safe region};
\try
\State  calculate $\delta_i^\star $ by  solving \eqref{eq:maximal internal input set} with safe region $\myset{Q}_i^\prime$;
\State compute the corresp. iAGC $(  I_{\underline{W}_i^\star}, I_{\underline{X}_i},  I_{\underline{Y}_i} )$ 
\catch{infeasible}
\State \textbf{return False};
 \endtry
\State update $ \mathcal{I}_{0},  \mathcal{I}_{1},  \mathcal{I}_{-1}$ by Alg. \ref{alg:update I_0 and I_1}.
\EndFor
\EndWhile
\State \textbf{return True}.
\end{algorithmic}
\end{algorithm}

\begin{algorithm}[h]
\caption{\texttt{Update safe region}} \label{alg:update safe region}
\begin{algorithmic}[1]
\Require Safe region $\myset{Q}_i $ and  iAGCs  $(  I_{\underline{W}_k}, I_{\underline{X}_k},  I_{\underline{Y}_k} )$ for all $k\in \text{Child}(i)$.
\State $M_i \leftarrow  \bigcap_{k\in \text{Child}(i)} o_i^{-1}(\text{Proj}_{i}(\underline{W}_{k})) \cap \myset{Q}_i$  
\State $\myset{Q}_i^\prime \leftarrow $ largest inner-approximation of $M_i$ by \eqref{eq:alg_safe_region}
\State \textbf{return} $\myset{Q}_i^\prime$.
\end{algorithmic}
\end{algorithm}

\begin{algorithm}[h]
\caption{\texttt{Update $\myset{I}_0,  \myset{I}_1  $ and $\mathcal{I}_{-1}$}} \label{alg:update I_0 and I_1}
\begin{algorithmic}[1]
\Require Subsystem $G_i $,  $\myset{I}_0,  \myset{I}_1  $ and $\mathcal{I}_{-1}$. 
\State $\mathcal{I}_0 \leftarrow \mathcal{I}_0 \setminus \{i\}, \mathcal{I}_{-1} \leftarrow \mathcal{I}_{-1} \cup \{i\}$,
\For{each subsystem $G_k$, $k \in \mathcal{I}_1$}
\If{$\text{Child}(k)\subseteq \mathcal{I}_{-1}$}
\State  $\mathcal{I}_0 \leftarrow \mathcal{I}_0 \cup \{k\}, \mathcal{I}_{1} \leftarrow \mathcal{I}_{1} \setminus\{k\}$,
\EndIf
\EndFor
\end{algorithmic}
\end{algorithm}

\vspace{1mm}
\subsubsection{Homogeneous interconnected system} In this case, we consider the homogeneous interconnected system in the following sense.
\begin{definition} \label{def:homo_sys}
    An  interconnected system $ \langle (G_i)_{i\in \myset{I}}, \mathcal{E}\rangle$ is called homogeneous if  $G_i = G_j$ and $ \myset{Q}_i = \myset{Q}_j, \forall i, j \in \mathcal{I}$.
\end{definition}

\begin{algorithm}[h]
\caption{\texttt{Contract construction for homogeneous systems}} \label{alg:homo_system}
 \algblock[TryCatchFinally]{try}{endtry}
\algcblock[TryCatchFinally]{TryCatchFinally}{finally}{endtry}
\algcblockdefx[TryCatchFinally]{TryCatchFinally}{catch}{endtry}
	[1]{\textbf{catch} #1}
	{\textbf{end try}}
 
\begin{algorithmic}[1]
\Require $G_i, \myset{Q}_i$
\try
\State  calculate $\delta_i^\star $ by solving \eqref{eq:maximal internal input set}
\State calculate $\zeta_i^\star$ by letting $\delta_i = \delta_i^\star$ and solving \eqref{eq:minimal safe set}
\State compute the corresp. iAGC $C_i =(  I_{\underline{W}_i^\star}, I_{\underline{X}_i^\star},  I_{\underline{Y}_i^\star} )$ 
\catch{infeasible}
\State \textbf{return False}
 \endtry
 \State Assign all subsystem $G_j, j\in \mathcal{I}$ with an iAGC $C_j =(  I_{\underline{W}_j^\star}, I_{\underline{X}_j^\star},  I_{\underline{Y}_j^\star} )$ with $\underline{W}_j^\star = \underline{W}_i^\star, \underline{X}_j^\star = \underline{X}_i^\star, \underline{Y}_j^\star = \underline{Y}_i^\star $.
 \If{ $\underline{Y}_j^\star \subseteq \text{Proj}_{j}(\underline{W}_{i}^\star) $ for all $j\in N(i)$ }
\State \textbf{return True}
\Else
\State $\myset{Q}_i^\prime \leftarrow$ update the local safe region $\myset{Q}_i$ by Alg. \ref{alg:update safe region};
\State Goto Step 1 with updated safe region $\myset{Q}_i^\prime $
\EndIf
\end{algorithmic}
\end{algorithm}

Algorithm \ref{alg:homo_system} starts with solving for one subsystem the maximal internal input set and the corresponding minimal safe region. If the compatibility condition is met, then we have verified the safety of the interconnected system; otherwise, we will reduce the safe region by taking the set intersection in Algorithm \ref{alg:update safe region} and start the same process with the updated safe region $\myset{Q}_{i}^\prime$.

We have the following results in this case:
\begin{proposition} \label{prop:homogeneous_case}
     Consider a homogeneous interconnected system as per Definition \ref{def:homo_sys}. Assume that $\{x_i: q_i(x_i) \geq a\} \subseteq X_i^0 $ for some $a>0$. Algorithm \ref{alg:homo_system} has the following properties:
     \begin{enumerate}
    \item Algorithm \ref{alg:homo_system} returns either \texttt{True} or \texttt{False} eventually.
    \item If Algorithm \ref{alg:homo_system} returns \texttt{True}, then  iAGCs $C_i= (  I_{\underline{W}_i^\star}, I_{\underline{X}_i^\star},  I_{\underline{Y}_i^\star} )$, $i\in \mathcal{I}$  satisfy the conditions in Proposition \ref{prop:iAGC_composition}.
    \item If Algorithm \ref{alg:homo_system} returns \texttt{False}, then there exists no common contract $  C_0 =(  I_{\underline{W}_0}, I_{\underline{X}_0},  I_{\underline{Y}_0} )$ such that $G_i \models C_0, i\in \mathcal{I}$ and  that the conditions in Proposition \ref{prop:iAGC_composition} are satisfied under Assumption \ref{ass:ass1}.
\end{enumerate}
\end{proposition}
\begin{proof}
    Consider the case when 
  $X_i^0$ is a subset of  the local safe region $\myset{Q}_i$ (otherwise,  \eqref{eq:maximal internal input set} yields \texttt{infeasible} in Step 2). There are three possibilities: 1) \eqref{eq:maximal internal input set} is infeasible, for which the algorithm terminates with \texttt{False}; 2) the contract compatibility condition is satisfied, for which the algorithm terminates with \texttt{True}; 3) the algorithm starts a new iteration. For the third scenario, the local safe region $\myset{Q}_i$ is updated at every iteration, and at each iteration, the set size gets smaller. In particular, as the set size is parameterized by a scaler $\zeta$, which is monotonically increasing, we know it has to converge to some number smaller than $a$ (in which case we found compatible local contracts and the algorithm terminates with \texttt{True}) or becomes larger than $a$ (in which case \eqref{eq:maximal internal input set} yields \texttt{infeasible}). In either case, the algorithm will return \texttt{True} or \texttt{False} if the iteration goes to infinite.

    If Algorithm \ref{alg:homo_system} returns \texttt{True}, from Step 9, the contract compatibility condition holds for subsystem $G_i$. Noting that each subsystem has the same number of parent nodes (implicit from the fact that each subsystem has the same output set $Y_i = Y_j$ and the same internal input set $W_i = W_j$) and the same assumption and guarantee sets (Step 8), we know this compatibility condition also holds for all subsystems. Thus, Claim 2) is shown.

    We show Claim 3) by contradiction. Suppose that a common local contract $C_0  =(  I_{\underline{W}_0}, I_{\underline{X}_0},  I_{\underline{Y}_0} )$ exists and the compatibility condition \eqref{eq:agreeing_condition} holds. Denote the safe region with which such an iAGC is obtained by $\myset{Q}_0$. That is, the functions defining these sets fulfill the SOS constraints in \eqref{eq:minimal safe set} (but not necessarily minimize the size of the safe region) and 
    \begin{equation} \label{eq:Q0_homo}
        o(\myset{Q}_0) \subseteq \text{Proj}_i(\underline{W}_0).
    \end{equation}   
    
    Trivially, we know $X_i^0 \subseteq \myset{Q}_0 \subseteq \myset{Q}_i$. Let the sequence of the updated safe regions of Algorithm \ref{alg:homo_system} be $Q^1 = \myset{Q}_i, Q^2, \ldots, Q^M$ (it is a finite sequence of sets with shrinking size following Claim 1)). Without loss of generality, assume $\myset{Q}^{r+1} \subsetneqq \myset{Q}_0 \subseteq \myset{Q}^{r} $. Following Proposition \ref{prop:local_sets} items 2 and 4, we know $\underline{W}_0 \subseteq\left.\underline{W}^\star\right\vert_{\myset{Q}_{0}} \subseteq \left.\underline{W}^\star\right\vert_{\myset{Q}^{r}} $, where $\left.\underline{W}^\star\right\vert_{\myset{Q}}$ represents the maximal internal input set given safe region $\myset{Q}$. Recall $Q^{r+1}$ is obtained from Alg. 2, i.e., $Q^{r+1}$ is the largest safe region such that $ o(\myset{Q}^{r+1}) \subseteq \text{Proj}_i(\left.\underline{W}^\star\right\vert_{\myset{Q}^r})$.
    However, from \eqref{eq:Q0_homo}, we know $ o(\myset{Q}_0) \subseteq \text{Proj}_i(\underline{W}_0) \subseteq \text{Proj}_i(\left.\underline{W}^\star\right\vert_{\myset{Q}^r})$ and $\myset{Q}^{r+1} \subsetneqq \myset{Q}_0$. This yields a contradiction, which thus proves Claim 3).
    \end{proof}

A common practice to bound the total number of iterations is to add extra termination conditions, e.g., Algorithm \ref{alg:homo_system} terminates if the updated safe region $\myset{Q}_i^\prime $ in Line 9 (with its level value $\zeta^\prime$) is close in size compared to the original one $\myset{Q}_i$ (with its level value $\zeta$), i.e., $ \zeta^\prime-\zeta<\epsilon$ for some small positive constant $\epsilon$. Other termination conditions include the maximal number of iterations allowed. When the algorithm is terminated due to these conditions, we do not have a definite conclusion about the existence of compatible contracts. 

\vspace{1mm}
\subsubsection{General case}
In the following, we provide a constructive approach for safety verification of interconnected systems with general connectivity graphs and dynamics. 

\begin{algorithm}[h]
\caption{\texttt{Contract construction for general systems}} \label{alg:general_system}
 \algblock[TryCatchFinally]{try}{endtry}
\algcblock[TryCatchFinally]{TryCatchFinally}{finally}{endtry}
\algcblockdefx[TryCatchFinally]{TryCatchFinally}{catch}{endtry}
	[1]{\textbf{catch} #1}
	{\textbf{end try}}
 
\begin{algorithmic}[1]
\Require Subsystem $G_i$ and its safe region $\myset{Q}_i ,i \in \mathcal{I}$ and connectivity graph $(\myset{I},\myset{E})$ 
\State run Algorithm \ref{alg:acyclic_graph}; 
\For{ each $G_i \in \mathcal{I}_1$}
\try
\State  calculate $\delta_i^\star $ by solving \eqref{eq:maximal internal input set};
\State calculate $\zeta_i^\star$ by letting $\delta_i = \delta_i^\star$ and solving \eqref{eq:minimal safe set}
\State compute the corresp. iAGC $(  I_{\underline{W}_i^\star}, I_{\underline{X}_i^\star},  I_{\underline{Y}_i^\star} )$ 
\catch{infeasible}
\State \textbf{return False};
 \endtry
 \EndFor
 \If{  contract compatibility condition \eqref{eq:agreeing_condition} does not hold}
 \State update $\myset{Q}_i$ for all $i\in \mathcal{I}_1$ by Alg. \ref{alg:update safe region}
 \State Goto Step 2
\EndIf
\State \textbf{return True}.
\end{algorithmic}
\end{algorithm}

\begin{proposition} \label{prop:general_case}
    For the general case, if Algorithm \ref{alg:general_system} returns \texttt{True}, then  iAGCs $C_i = (  I_{\underline{W}_i}, I_{\underline{X}_i},  I_{\underline{Y}_i} ), i\in \mathcal{I}$ satisfies the conditions in Proposition \ref{prop:iAGC_composition}.
\end{proposition}

The proof is straightforward and omitted. We note that Algorithm \ref{alg:general_system} simplifies to Algorithm \ref{alg:acyclic_graph} in the case of acyclic connectivity graph, and becomes Algorithm \ref{alg:homo_system} for homogeneous interconnected systems. Proposition \ref{prop:general_case} only provides sufficient conditions on the existence of compatible contracts. We will explore how to provide completeness guarantees for Algorithm \ref{alg:general_system} or design other algorithms with such guarantees in our future work.

\begin{remark}[Meaning of \texttt{False}]
    It is worth highlighting that our completeness results are established under Assumption \ref{ass:ass1}, and that the algorithm returning \texttt{False} does not mean that the interconnected system is unsafe. It simply means that we can not certify the safety of the interconnected systems under the restrictions in Assumption \ref{ass:ass1}. When a \texttt{False} is encountered, one can group two or several subsystems together and run the algorithms again by taking a group of subsystems as a single one.
\end{remark}

\section{Examples} \label{sec:examples}
\subsection{Vehicular platooning: an acyclic example }

Consider a vehicular platooning scenario adapted from~\cite{liu2019compositional} where $N+1$ autonomous vehicles are moving on a single-lane road. We assume that each vehicle has the same length $l$, and has access to the state information of its proceeding vehicle and the leading vehicle (hereafter referred to as the leader). The dynamics relative to the leader is
\begin{equation} \label{eq:platooning_dyn}
    \begin{aligned}
        \tildedot{p}_i & = \tilde{v}_i, \\
        \tildedot{v}_i & = \tilde{u}_i  -(\tilde{v}_i - \tilde{v}_{i-1})^3,
    \end{aligned}
\end{equation}
where $\tilde{p}_i(t) = p_0(t) - p_i(t), \tilde{v}_i = v_0(t) - v_i(t), \tilde{u}_i(t) = u_0(t) - u_i(t)$, $i = 1,2,\ldots,N, k_0 >0$. Here $ p_i, v_i, u_i, \tilde{p}_i, \tilde{v}_i, \tilde{u}_i$ denote the absolute position, the absolute velocity, the absolute control, the relative position, the relative velocity, and the relative control of vehicle $i$, respectively. We conveniently let $\tilde{v}_0 = 0$ for ease of notation. 

 \begin{figure}[t]
    \centering
    \includegraphics[width=0.9\linewidth]{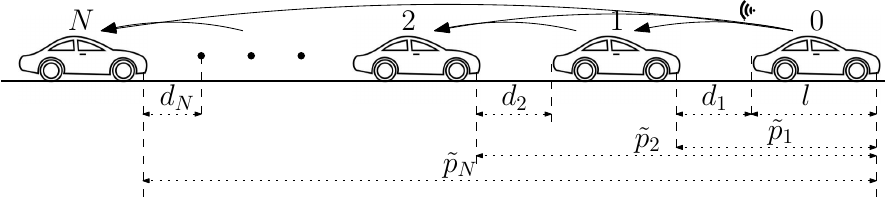}
    \caption{Platooning scenario}
    \label{fig:platooning}
\end{figure}

\begin{figure*}[t]
	\centering
	\begin{subfigure}[t]{0.24\linewidth}
		\includegraphics[width=\linewidth]{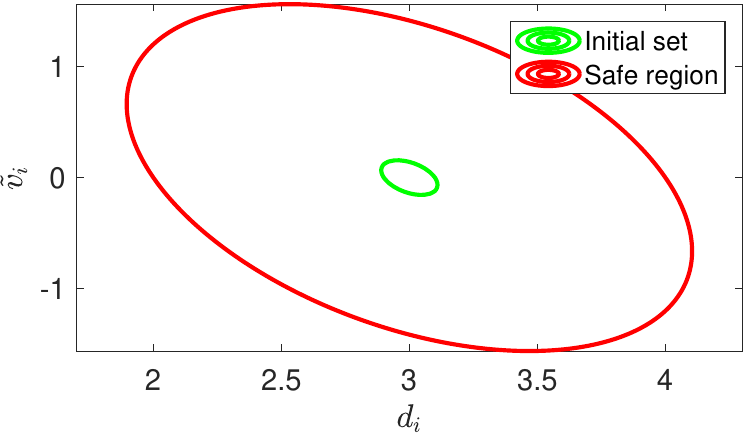}
		\caption{  $X_i^0$ and $\myset{Q}_i$, $i = 1,2,3$. }   
	\end{subfigure} 
	\begin{subfigure}[t]{0.24\linewidth}
		\includegraphics[width=\linewidth]{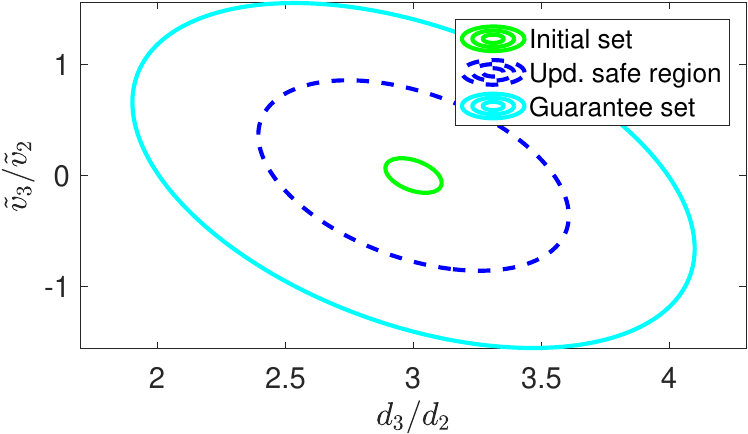}
		\caption{  AGC sets for vehicle $3$.}
	\end{subfigure}
	\begin{subfigure}[t]{0.24\linewidth}
		 \includegraphics[width= \linewidth]{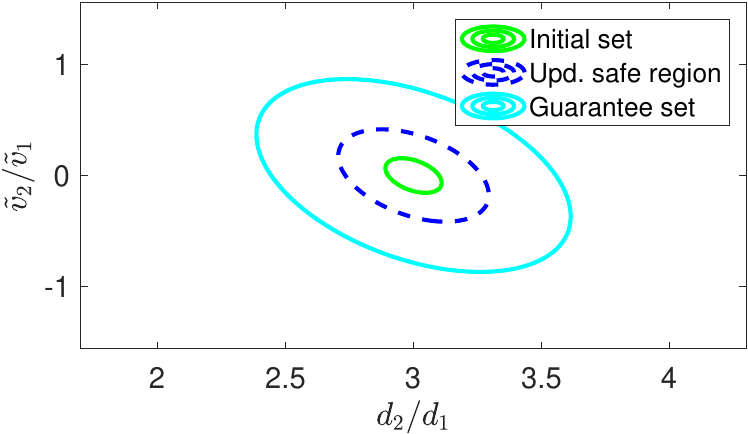}
		\caption{  AGC sets for vehicle $2$.}
	\end{subfigure}
 \begin{subfigure}[t]{0.24\linewidth}
		 \includegraphics[width= \linewidth]{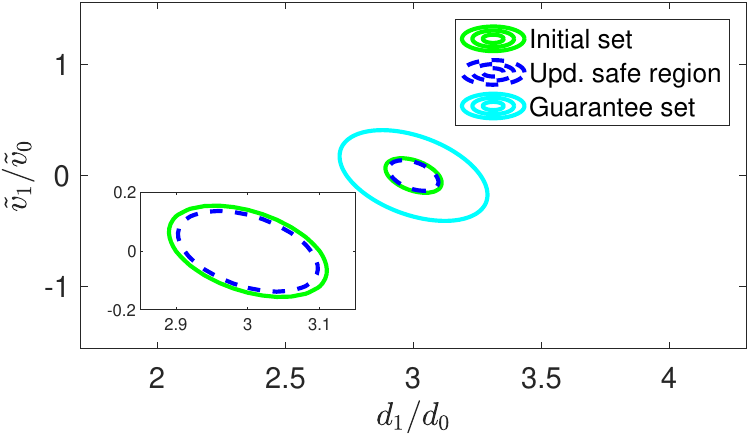}
		\caption{  AGC sets for vehicle $1$.}
 \end{subfigure}
	\setlength{\belowcaptionskip}{-18pt}
	\caption{ Results for the platooning example.}
	\label{fig:platooning_result}	
	\end{figure*}
 
Instead of looking at the relative dynamics in \eqref{eq:platooning_dyn}, we introduce a new coordinate $x_i = (d_i,\tilde{v}_i)$ associated with vehicle $i$, where $d_i = \tilde{p}_i - \tilde{p}_{i-1} - l$ denotes the distance between the front of the $i$-th vehicle and the rear of its proceeding vehicle. See Figure \ref{fig:platooning} for an illustration. The dynamics of this new state are given by
\begin{equation} \label{eq:platooning_new_dyn}
    \begin{aligned}
        \dot{d}_i & = \tilde{v}_i - \tilde{v}_{i-1} \\
        \tildedot{v}_i & =  - (\tilde{v}_i - \tilde{v}_{i-1})^3 + \tilde{u}_i 
    \end{aligned}
\end{equation}

From the analysis above, we can model the vehicular platooning system as an interconnected system. Each subsystem $G_i, i\in \mathcal{I}=\{1,2,\ldots,N\},$ is a continuous-time system $$G_i = (U_i, W_i, X_i, Y_i, X_i^0,\mathcal{T}_i)$$ with $(d_i,  \tilde{v}_i)$ as the state vector, and $U_i =\mathbb{R},  W_i = \left\{\begin{smallmatrix}
    \emptyset, & i = 1; \\
    \mathbb{R}, & i\geq 2,
 \end{smallmatrix}\right.$, $X_i =  \mathbb{R}^2, Y_i =  \left\{\begin{smallmatrix}
    \mathbb{R}, & i \leq N-1; \\
    \emptyset, &  i= N,
 \end{smallmatrix}\right., X_i^0 \subset X_i, \mathcal{T}_i$ characterizes all the trajectories satisfying \eqref{eq:platooning_new_dyn}, and the output map $o_i: \left\{\begin{smallmatrix}(d_i, \tilde{v}_i) \mapsto \tilde{v}_i & i\leq N-1 \\
 (d_i, \tilde{v}_i) \mapsto \emptyset & i = N
 \end{smallmatrix}\right.$. The binary connectivity relation $\mathcal{E} $ is defined that $(j,i)\in \mathcal{E}$ if and only if $j = i-1, i = 2,3,\ldots, N$. One verifies that $\{G_i\}_{i\in \mathcal{I}}$ is compatible for composition with respect to $\mathcal{E} $.  In the following analysis, we assume all the subsystems have the same initial state set $X_i^0 $ and safe region $\myset{Q}_i $, which are $X_i^0 = \{x_i: -x_i^\top Q x_i + q^\top x_i - 899 \geq 0\},$ and the safe region $\myset{Q}_i = \{ x_i: - x_i^\top Q x_i + q^\top x_i - 800 \geq 0\}.$ where $Q = \begin{bsmallmatrix}
     100 & 30 \\
     30 & 50
 \end{bsmallmatrix}$ and $q = (600, 180)$. The initial state set and the safe region are depicted in Fig.\ref{fig:platooning_result}(a). Each subsystem applies a local controller $$\tilde{u}_i =  - (\tilde{v}_i - \tilde{v}_{i-1}) -(d_i - 3) -(d_i - 3)^3, i\in \mathcal{I}. $$
 
 Our task is to verify safety of the interconnected system $ \langle (G_i)_{i\in \myset{I}}, \mathcal{E}\rangle$. In this scenario, we consider $4$ vehicles ($N = 3$).
Algorithm \ref{alg:acyclic_graph} will be used for this example. We will fix the form of $d^i_{j_k}$ in $\text{Proj}_{k}(W_{i})  $ as in $d^i_{j_k}(x_{i-1}) =  a^2 - (\tilde{v}_{i-1} - b)^2 $ to denote a bounded interval, where $a $ and $b$ are determined accordingly.

Consider the vehicle $3$, which is the leaf node in the graph. One calculates $d^3_2(x_{2}) = 2.439 -\tilde{v}_2^2 $ by projecting $\myset{Q}_2$ to the $\tilde{v}_2$ coordinate. By solving \eqref{eq:maximal internal input set} and \eqref{eq:minimal safe set}, one obtains $\delta ^\star =  1.704$ and $\zeta^\star = 1.1147$. That is to say, a local contract $C_3 = ( I_{\underline{W}_3}, I_{\underline{X}_3}, I_{\underline{Y}_3})$ for vehicle $3$ is constructed with
\begin{subequations}
    \begin{align}
        \underline{W}_3 & = \{ \tilde{v}_{2}: 0.735 -\tilde{v}_{2}^2 \geq 0 \}, \\
        \underline{X}_3 & =  \{ x_3: -x_3^\top Q x_3 + q^\top x_3 - 801.115 \geq 0\} \},  \\
        \underline{Y}_3 & =  o_3(\underline{X}_3).
    \end{align}
\end{subequations}
$\underline{W}_3$ can be seen as the requirement from vehicle $3$ to vehicle $2$. Following Algorithm \ref{alg:update safe region} (and also by solving \eqref{eq:alg_safe_region}), the updated safe region for vehicle $2$ is
$ \myset{Q}_2^\prime = \{ x_2:  -x_2^\top Q x_2 + q^\top x_2 -  869.85 \geq 0\}$. The initial set and the guarantee set of vehicle $3$ as well as the updated safe region of vehicle $2$ are illustrated in Figure \ref{fig:platooning_result}(b). We thus follow the same procedures for vehicles $2$ and $1$, and obtain the results in Figure \ref{fig:platooning_result}(c) and Figure \ref{fig:platooning_result}(d). Moreover, we calculate the assumption set $ \underline{W}_1  = \{ \tilde{v}_{0}: 0.019 -\tilde{v}_{0}^2 \geq 0 \},$ which holds true since $\tilde{v}_{0}(t) = 0$ for all $t$. Thus, we have constructed compatible local contracts for each vehicle. Following Proposition \ref{prop:iAGC_composition}, we conclude that the platooning system is safe.

\subsection{Room temperature: a homogeneous example}
In the second example, we consider a room temperature regulation problem \cite{girard2015safety} in a ring-shaped building as illustrated in Fig. \ref{fig:room_scenario}. Each room has its temperature $x_i$, which is affected by neighboring rooms, the heater, and the environment as follows
\begin{equation*}
    \begin{aligned}
        \dot{x}_{i}(t) & = \alpha(x_{i+1} + x_{i-1} - 2x_{i})+ \beta(t_e - x_i) + \gamma(t_h - x_i)u_i,\\
        y_i(t) & =  x_i,
    \end{aligned}
\end{equation*}
where $x_{i+1}, x_{i-1}$ are the temperatures of room $i+1$ and $i-1$ (and we conveniently let $x_0(t) = x_N(t), x_{N+1}(t) = x_1(t)$), $t_e, t_h$ are the temperatures of the environment and the heater, respectively. $\alpha, \beta, \gamma$ are the respective conduction factors for the neighboring room, the environment, and the heater. $u_i$ denotes the valve control to the heater. Choose $(t_e, t_h, \alpha, \beta, \gamma) = (-1, 50, 0.05, 0.008, 0.004) $, and $$u_i = 0.05(x_{i+1} + x_{i-1} - 2x_i) + 0.05(25 -x_i).$$ The  initial set is $\myset{S}_{I,i} = [24,26]$ and the safe region is $\myset{Q}_i = [20,30]$ for every room. 

We can model the temperature system as an interconnected system. In particular, each subsystem $G_i = (U_i, W_i, X_i, Y_i, X_i^0,\mathcal{T}_i)$ has $x_i$ as the state, $(x_{i-1},x_{i+1})$ as the internal input, $u_i$ as the external input, $o_i(x_i) = x_i$, $U_i , X_i, Y_i= \mathbb{R}, W_i =  \mathbb{R}^2, X_i^0 = \{x_i: 1 - (x_i - 25)^2 \geq 0\} $, and $\myset{Q}_i = \{x_i: 5^2 - (x_i - 25)^2 \geq 0\}$. The connectivity relation $\myset{E}$ is defined that $(j,i)\in \myset{E}$ if and only if $j = i\pm 1, i = 1,2,\ldots, N$. Per Definition \ref{def:homo_sys}, this is a homogeneous interconnected system and we will apply Algorithm \ref{alg:homo_system} for this example.

At the first iteration, by solving \eqref{eq:maximal internal input set} and \eqref{eq:minimal safe set}, we obtain $\delta^\star = 20.575, \zeta^\star = 0$. Thus, we have constructed a local iAGC $C_i = ( I_{\underline{W}_i}, I_{\underline{X}_i}, I_{\underline{Y}_i}) $ with 
\begin{equation*}
    \begin{aligned}
        &  \underline{W}_i = \{ (x_{i-1}, x_{i+1}): -x_{j}^2 + 50x_{j} - 620.575\geq 0, j = i\pm 1 \}, \\
      &  \underline{X}_i = \underline{Y}_i =\{ x_i: -x_i^2 + 50x_i - 600\geq 0 \}.
    \end{aligned}
\end{equation*}
After assigning the same local contract to all subsystems, one verifies that the contract compatibility condition \eqref{eq:agreeing_condition} does not hold. According to Step 12 of Algorithm \ref{alg:homo_system}, we update the safe region for each room to be $\myset{Q}_i^\prime =  \{ x_{i}: -x_{i}^2 + 50x_{i} - - 620.575\geq 0\}$ and start over. For the second iteration, we obtain local iAGC $C_i = ( I_{\underline{W}_i}, I_{\underline{X}_i}, I_{\underline{Y}_i}) $ with
\begin{equation*}
    \begin{aligned}
        &  \underline{W}_i = \{ (x_{i-1}, x_{i+1}): -x_{j}^2 + 50x_{j} - 622.138\geq 0, j = i\pm 1 \}, \\
      &  \underline{X}_i = \underline{Y}_i =\{ x_i: -x_i^2 + 50x_i - 623.575\geq 0 \}.
    \end{aligned}
\end{equation*}
This time, one verifies that the compatibility condition \eqref{eq:agreeing_condition} holds, and thus, certifies the safety of the room temperature system. An illustration of the assume and the guarantee sets is given in Fig. \ref{fig:room_result}. We note that the computation expense is not related to the number of rooms $N$, and only small-size SOS optimization problems involving $3$ independent variables are to be solved. This is in contrast to a naive SOS approach for synthesizing a barrier function, which will become intractable when thousands of rooms are involved.

\begin{figure}
    \centering
    \includegraphics[width=0.45\linewidth]{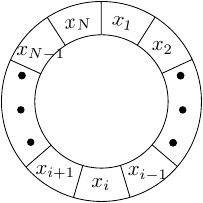}
    \caption{ Room temperature scenario.}
    \label{fig:room_scenario}
\end{figure}

\begin{figure}
    \centering
    \centering
	\begin{subfigure}[t]{0.45\linewidth}
 \includegraphics[width=\linewidth]{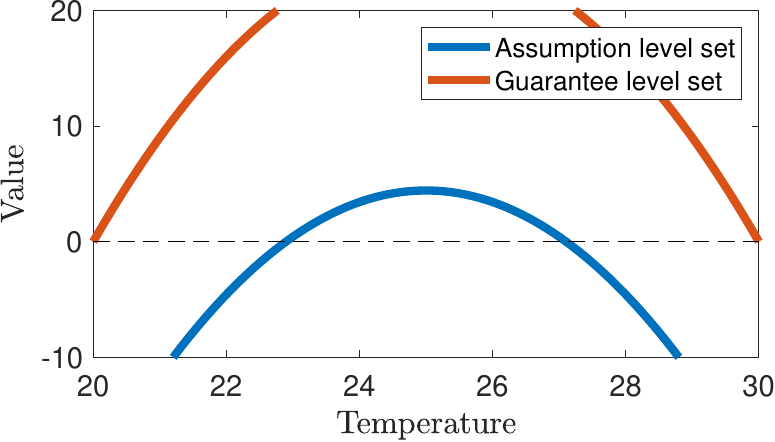}
 \end{subfigure}
 \begin{subfigure}[t]{0.45\linewidth}
    \includegraphics[width=\linewidth]{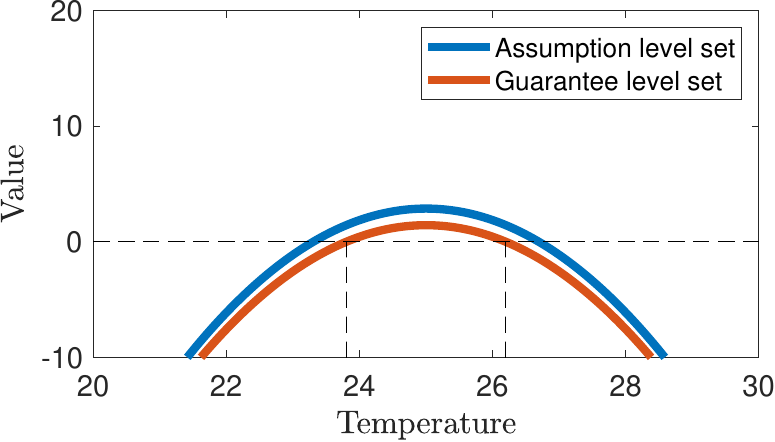}
 \end{subfigure}
    \caption{ Assume/guarantee sets for the room temperature example. Left: iteration 1, right: iteration 2.}
    \label{fig:room_result}
\end{figure}

\section{Conclusions} \label{sec:conclusion}
In this work, we propose a safety verification scheme for interconnected continuous-time nonlinear systems based on assume-guarantee contracts (AGCs) and sum-of-squares (SOS) programs. The proposed scheme uses SOS optimization to calculate local invariance AGCs by synthesizing local (control) barrier functions, and then negotiates among neighboring subsystems at the contract level. If the proposed algorithms find compatible local contracts, safety property of the interconnected system is certified. We also show that the algorithms will terminate in finite steps and will always find a solution when one exists in the case of acyclic connectivity graphs or for homogeneous systems. We also demonstrate the effectiveness of the proposed algorithms for vehicle platooning and room temperature regulation examples.

\bibliographystyle{IEEEtran}
\bibliography{IEEEabrv,references}

\end{document}